\documentclass[11pt,reqno]{amsart}
\usepackage{url,framed,todonotes}

\usepackage{latexsym}
\usepackage{dsfont}

\usepackage{mathrsfs}
\usepackage[colorlinks]{hyperref}
\usepackage[square,authoryear]{natbib}

\usepackage{amsfonts}
\usepackage{amsmath}
\usepackage{amscd}
\usepackage{amssymb}
\usepackage{geometry}
\usepackage{color}
\usepackage{graphicx}%
\def\bx{{\mathbf {x} }}

\def\bsxi{{\boldsymbol {\xi} }}

\def\dd{{\boldsymbol d}}
\newcommand{\bfi}{\bfseries\itshape}
\newcommand{\rem}[1]{}

\DeclareMathOperator{\ad}{ad}

\newcommand{\lform}[2]{{\big\langle {#1}\, , \, {#2}\big\rangle}}
\newcommand{\Lform}[2]{{\Big\langle {#1}\, , \, {#2}\Big\rangle}}

\newcommand{\id}{{\mathrm{id}}}
\newcommand{\ti}{\times}
\newcommand{\Ad}{\text{Ad}}
\renewcommand{\ad}{\text{ad}}
\newcommand{\de}{\delta}
\newcommand{\om}{\omega}
\newcommand{\al}{\alpha}

\providecommand{\U}[1]{\protect\rule{.1in}{.1in}}
\newtheorem{theorem}{Theorem}

\newtheorem{corollary}[theorem]{Corollary}

\newtheorem{definition}[theorem]{Definition}

\newtheorem{proposition}[theorem]{Proposition}
\newtheorem{remark}[theorem]{Remark}

\makeatletter
\@addtoreset{figure}{section}
\def\thefigure{\thesection.\@arabic\c@figure}
\def\fps@figure{h, t}
\@addtoreset{table}{bsection}
\def\thetable{\thesection.\@arabic\c@table}
\def\fps@table{h, t}
\@addtoreset{equation}{section}

\makeatother

\begin{document}

\title{Stochastic Metamorphosis in Imaging Science}
\author{Darryl D. Holm,\\ \\
In honour of David Mumford on his 80$^{\rm th}$ birthday!}

\address{DDH: Department of Mathematics, Imperial College, London SW7 2AZ, UK.}

\date{\today}

\begin{abstract}
In the pattern matching approach to imaging science, the process of \emph{metamorphosis} in template matching with dynamical templates was introduced in  \cite{ty05b}. 
In \cite{HoTrYo2009} the metamorphosis equations of \cite{ty05b} were recast into the Euler-Poincar\'e variational framework of \cite{HoMaRa1998} and shown to contain the
equations for a perfect complex fluid \cite{Holm2002}. This result related the data structure underlying the process of metamorphosis in image matching to the physical concept of order parameter in the theory of complex fluids \cite{GBHR2013}. In particular, it cast the concept of Lagrangian paths in imaging science as deterministically evolving curves in the space of diffeomorphisms acting on image data structure, expressed in Eulerian space. (In contrast, the landmarks in the standard LDDMM approach are Lagrangian.)

For the sake of introducing an Eulerian uncertainty quantification approach in imaging science, we extend the method of metamorphosis to apply to image matching along \emph{stochastically} evolving time dependent curves on the space of diffeomorphisms. The approach will be guided by recent progress in developing stochastic Lie transport models for uncertainty quantification in fluid dynamics in \cite{holm2015variational,CrFlHo2017}.

\end{abstract}

\maketitle

\tableofcontents

\eject

\section{Introduction}

In recent work \cite{AHS2017,arnaudon2016stochastic2}, a new method of modelling variability of shapes has been introduced. This method uses a theory of stochastic perturbations consistent with the geometry of the diffeomorphism group corresponding to the Large Deformation Diffeomorphic Metric Mapping framework (LDDMM, see \cite{younes_shapes_2010}).  In particular, the method introduces stochastic Lie transport along stochastic curves in the diffeomorphism group of smooth invertible transformations. It models the development of variability as observed, for example, when human organs are influenced by disease processes, as analysed in computational anatomy \cite{younes_evolutions_2009}. It also provides a framework including a Hamiltonian formulation for quantifying uncertainty in the development of shape atlases in computational anatomy. Hamiltonian methods for deterministic computational anatomy were recently reviewed in \cite{MTY2015}.

The theory developed in \cite{AHS2017,arnaudon2016stochastic2} treats LDDMM as a flow and uses methods based on stochastic fluid dynamics introduced in \cite{holm2015variational}. It addresses the problem of uncertainty quantification by introducing spatially correlated noise which respects the geometric structure of the data. 
Thus, the method provides a new way of characterising stochastic variability of shapes using spatially correlated noise in the context of the standard LDDMM framework. Numerical methods for addressing stochastic variability of shapes with landmark data structure have also been developed in  \cite{holm2016stochastic,HoTy2016,AHS2017,arnaudon2016stochastic2}.

Although the examples were limited to landmark dynamics in the work \cite{AHS2017,arnaudon2016stochastic2}, it was clear that Lie-transport noise can be applied to any of the data structures used in the LDDMM framework, because it is compatible with the transformation theory on which LDDMM is based. 
The LDDMM theory was initiated by \cite{trouve_infinite_1995,christensen_deformable_1996,dupuis_variational_1998,MiTrYo2002,beg2005computing} building on the pattern theory of \cite{grenander_general_1994}. The LDDMM approach models shape comparison (registration) as dynamical transformations from one shape to another whose data structure is defined as a tensor valued smooth embedding. These shape transformations are expressed in terms of the action of diffeomorphic flows, regarded as time dependent curves of smooth transformations of shape spaces. 
This provides a unified approach to shape modelling and shape analysis, valid for a range of structures such as landmarks, curves, surfaces, images, densities and tensor-valued images. 
For any such data structure, the optimal shape deformations are described via the Euler-Poincar\'e equation of the diffeomorphism group, usually referred to as the EPDiff equation \cite{holm1998euler,HoMa2004,younes_evolutions_2009}. 
The work \cite{AHS2017,arnaudon2016stochastic2} showed how to obtain a stochastic
EPDiff equation valid for any data structure, and in particular for the finite dimensional representation of images based on landmarks.
For this purpose, the work \cite{AHS2017,arnaudon2016stochastic2} followed the Euler-Poincar\'e derivation of LDDMM of \cite{bruveris2011momentum} based on geometric mechanics \cite{MR1994,holm_geometric_2011} and the use of momentum maps to represent images and shapes.
The introduction of Lie-transport noise into the EPDiff equation was implemented as cylindrical noise, obtained by pairing the deterministic momentum map the sum over eigenvectors of the spatial covariance of Stratonovich noise, each with its own Brownian motion.

The work \cite{AHS2017,arnaudon2016stochastic2} was not the first to consider stochastic evolutions in LDDMM.
Indeed, \cite{TrVi2012,vialard2013extension} and more recently \cite{marsland2016langevin} had already investigated the possibility of stochastic perturbations of landmark dynamics. 
In the earlier works, the noise was introduced into the landmark momentum equations, as though it were an external random force acting on each landmark independently. 
In \cite{marsland2016langevin}, an extra dissipative force was added to balance the energy input from the noise and to make the dynamics correspond to a certain type of heat bath used in statistical physics.
In contrast, the work \cite{AHS2017,arnaudon2016stochastic2} instead introduced Eulerian Stratonovich noise  into the reconstruction relation used to find the deformation flows from the action of the velocity vector fields on their corresponding momenta, which are solutions of the EPDiff equation \cite{HoMa2004,younes_shapes_2010}. 

As shown in \cite{AHS2017,arnaudon2016stochastic2}, this derivation of stochastic models is compatible with the Euler-Poincar\'e constrained variational principles, it preserves the momentum map structure and yields a stochastic EPDiff equation with a novel type of multiplicative noise, depending on both the gradient and the magnitude of the solution. 
The model in \cite{AHS2017,arnaudon2016stochastic2} was based on the previous works \cite{holm2015variational, arnaudon2016noise}, where, respectively, stochastic perturbations of infinite and finite dimensional mechanical systems were considered. 
The Eulerian nature of this type of noise implies that the noise correlation depends on the image position and not, as for example in \cite{TrVi2012,marsland2016langevin}, on the landmarks themselves.  
This property explains why the noise is compatible with any data structure while retaining the freedom in the choice of its spatial correlation.  

The present work extends the Euler-Poincar\'e variational framework for the metamorphosis approach of \cite{ty05b,HoTrYo2009} from the deterministic setting to the stochastic setting. Section \ref{determ-met-review} reviews the derivation of the deterministic metamorphosis equations as cast by \cite{HoTrYo2009} into the Euler-Poincar\'e variational framework of \cite{HoMaRa1998}, as well as several other developments, including the Hamilton-Pontryagin principle and two different  Hamiltonian formulations of deterministic metamorphosis. Section \ref{stoch-met-review} introduces metamorphosis by stochastic Lie transport and traces out its preservation and modification of the deterministic mathematical structures reviewed in Section \ref{determ-met-review}.

Thus, for the sake of potential applications in uncertainty quantification, this paper extends the method of metamorphosis for image registration to enable its application  to image matching along stochastically time dependent curves on the space of diffeomorphisms.

\section{Review of metamorphosis by deterministic Lie transport}  \label{determ-met-review}

In the pattern matching approach to imaging science, the process of ``metamorphosis'' in template matching with dynamical templates was introduced in  \cite{ty05b}. 
In \cite{HoTrYo2009} the metamorphosis equations of \cite{ty05b} were recast into the Euler-Poincar\'e variational framework of \cite{HoMaRa1998} and shown to contain the
equations for a perfect complex fluid \cite{Holm2002}. This result connected the data structure underlying the process of metamorphosis in image matching to the physical concept of order parameter in the theory of complex fluids. After developing the
general theory in \cite{HoTrYo2009}, various examples were reinterpreted, including point set, image and density metamorphosis. Finally, the issue was discussed of matching measures with metamorphosis, for which existence theorems for the initial and boundary value problems were provided. For more recent developments the the metamorphosis equations as well as numerical methods especially designed for metamorphosis, see \cite{RY2016}.

Let $N$ be manifold, which is acted upon by a Lie
group $G$. The manifold $N$ contains what we can refer to as ``deformable
objects'' and $G$ is the group of deformations, which is the group of
diffeomorphisms acting on the manifold $N$ in our applications. Several examples for the space $N$ were developed in the Euler-Poincar\'e context in \cite{HoTrYo2009}.

\begin{definition}
A {\bfi
metamorphosis} \cite{ty05b} is a pair of curves $(g_t,\,\eta_t)\in G \times N$ parameterized by time $t$,
with $g_0 = \id$. Its {\bfi image} is the curve $n_t\in N$ defined by the action 
$n_t = g_t.\eta_t$, where subscript $t$ indicates explicit dependence on time, $t$. 
The quantities $g_t$ and $\eta_t$ are called, respectively, the {\bfi deformation part} of the metamorphosis, and its {\bfi template part}. When $\eta_t$ is
constant, the metamorphosis is said to be a {\bfi pure
deformation}. In the general case, the image is a combination of a
deformation and template variation.
\end{definition}

Before introducing stochasticity, the next section provides notation and definitions for the general problem of metamorphoses in the deterministic case. Several derivations of the fundamental metamorphosis equations are given, in order to explore the various formulations of the problem from different perspectives. 


\subsection{Notation and Euler-Poincar\'e theorem for the deterministic case}
Following \cite{ty05b,HoTrYo2009}, we will use either letters $\eta$ or $n$ to denote elements of $N$,
the former being associated to the template part of a metamorphosis,
and the latter to its image. 

The variational problem we shall study optimizes over metamorphoses
$(g_t, \eta_t)$ by minimizing, for some Lagrangian $L:TG\times  TN \to \mathbb{R}$, the action integral
\begin{equation}
\label{eq:lag}
S = \int_0^1 L(g_t, \dot g_t, \eta_t, \dot\eta_t) dt\,,
\end{equation}
with fixed endpoint conditions for the initial and final images $n_0$
and $n_1$ (with $n_t = g_t\eta_t$) and $g_0 = \id_G$. That is, the
images $n_t$ are constrained at the end-points, with the initial condition $g_0 = \id$.

Let $\mathfrak
g$ denote the Lie algebra of the Lie group $G$. We will consider  Lagrangians
defined on $TG\times TN$, that satisfy the following invariance 
conditions: there exists a function $\ell$ defined on $\mathfrak g \ti
TN$ such that
\begin{equation}
\label{eq:lag-inv}
L(g, U_g, \eta, \xi_\eta) = \ell(U_gg^{-1}, g\eta, g\xi_\eta).
\end{equation}
In other words, $L$ is invariant under the right action of $G$ on $G\times
N$ defined by $(g, \eta)h = (gh, h^{-1}\eta)$. 

For a metamorphosis $(g_t, \eta_t)$, we therefore have a reduced Lagrangian, upon defining $u_t = \dot g_t
g_t^{-1}$, $n_t = g_t\eta_t$ and $\nu_t = g_t \dot \eta_t$, given by
\begin{equation}
\label{eq:reduced-lag}
L(g_t, \dot g_t, \eta_t, \dot\eta_t) = \ell(u_t, n_t, \nu_t).
\end{equation}

\bigskip

The Lie derivative with respect to a vector field $X$ will be denoted
$\mathcal L_X$. The Lie algebra of $G$ is identified with the set of
right invariant vector fields $U_g = ug$, $u\in T_\id G= \mathfrak g$, $g\in G$,
and we will use the notation $\mathcal L_u = \mathcal L_U$. 

The Lie bracket $[u,v]$ on the Lie algebra of smooth vector fields $\mathfrak g$ is defined by
\begin{equation}
\label{eq:Lie-brkt}
\mathcal L_{[u, v]} = -(\mathcal L_u\mathcal L_v - \mathcal
L_v\mathcal L_u)
\end{equation}
and the associated adjoint operator is $\ad_u v = [u,v]$.
Letting $I_g(h) = ghg^{-1}$ and $Ad_vg = \mathcal L_v I_g(\id)$, we
also have $\ad_u v = \mathcal L_u(\Ad_v)(\id)$. When $G$ is a group of
diffeomorphisms, this yields $ad_u v = du\, v -dv\, u$. 

The pairing between a linear form $\mu$ and a vector field $u$ will be denoted
$\lform{\mu}{u}$. Duality with respect to this pairing will be denoted
with an asterisk. For example, $N^*$ is the dual space of the manifold $N$ with respect to this pairing. 

When $G$ acts on a manifold $\tilde N$, the diamond  operator $(\diamond)$ is defined on $\tilde N^* \times \tilde N$ and takes values in the dual Lie algebra $\mathfrak g^*$. That is, $\diamond: \tilde N^* \times \tilde N \to \mathfrak g^*$.
For $\tilde n^*\in \tilde N^*$ and $\tilde n\in \tilde N$ the diamond operation is defined by 
\begin{equation}
\label{eq:diamond-def}
\lform{\tilde n^* \diamond \tilde n}{u}_{\mathfrak g} 
:= - \lform{\tilde n^*}{u\tilde n}_{T\tilde N}\,,
\end{equation}
where the action of a vector field $u\in \mathfrak g$ on $\tilde n\in \tilde N$ is denoted by simple concatenation, $u\tilde n\in T\tilde N$. For example, the Lie algebra action of the vector field $u\in \mathfrak g$ on $\tilde n\in \tilde N$ is denoted $u\tilde n = \mathcal{L}_u\tilde n \in T\tilde N$. Subscripts on the pairings in the definition \eqref{eq:diamond-def} indicate, as follows, $\lform{\,\cdot}{\cdot\,}_{\mathfrak g}: {\mathfrak g}^*\times {\mathfrak g} \to \mathbb{R}$ and $\lform{\,\cdot}{\cdot\,}_{T\tilde N}: {\tilde N}^*\times {T\tilde N} \to \mathbb{R}$. In what follows, for brevity of notation we will often suppress these subscripts $t$, except where we wish to emphasise the presence or absence of explicit time dependence.  Suppressing these subscripts when explicit time dependence is understood should cause no confusion. 
\bigskip

\begin{theorem}[Euler-Poincar\'e theorem]\label{equiv-thm-det}
With the preceding notation, the following four statements are equivalent for a metamorphosis Lagrangian $L$ that is invariant under the right action of $G$ on $G\times N$ defined by $(g, \eta)h = (gh, h^{-1}\eta)$, with fixed endpoint conditions for the initial and final images $n_0$ and $n_1$:
\begin{enumerate}
\item [{\bf i} ] 
Hamilton's variational principle 
\begin{equation} \label{hamiltonprinciple}
\delta S = \delta \int _{0} ^{1} L(g_t, \dot g_t, \eta_t, \dot\eta_t) dt = 0
\quad\hbox{for}\quad
L(g_t, \dot g_t, \eta_t, \dot\eta_t) = \ell(\dot{g}_tg_t^{-1}, g_t\eta_t, g_t\dot{\eta}_t)
\end{equation}
holds, for variations $\delta g_t$
of $ g_t $ and $\delta \eta_t$
of $ \eta_t $ vanishing at the endpoints.
\item [{\bf ii}  ] $g_t$ and $\eta_t$ satisfy the Euler--Lagrange
equations for $L$ on $TG\times TN$.
\item [{\bf iii} ]  The constrained variational principle
\begin{equation} \label{variationalprinciple}
\delta \mathcal{S} = \delta \int_0^1 \ell(u_t, n_t, \nu_t) dt = 0
\end{equation}
holds for Lagrangian $\ell$ defined on $\mathfrak g \ti
TN$ using variations of $u_t = \dot g_t
g_t^{-1}$, $n_t = g_t\eta_t$ and $\nu_t = g_t \dot \eta_t$ of the form
\begin{equation} \label{epvariations}
\delta u = \dot\xi_t - \ad_{u_t}\xi_t
, \quad
\delta n = \varpi_t + \xi_t n_t 
,\quad\hbox{and}\quad
\delta \nu = \dot\varpi_t + \xi_t\nu_t - u_t\varpi_t
,
\end{equation}
where $\xi_t  = \de g_t g_t^{-1}$, $\varpi_t = g_t\de\eta_t$ and these variations vanish at the endpoints.
\item [{\bf iv}] The {\bfi Euler--Poincar\'{e}}
equations hold on $\mathfrak{g} \times TN$\medskip
\begin{align}
\label{eq:meta.2-thm}
\begin{split}
&\frac{\partial}{\partial t} \frac{\delta \ell}{\delta u} + \ad^*_{u_t}
\frac{\delta \ell}{\delta u} + \frac{\delta \ell}{\delta n} \diamond n_t
+ \frac{\delta \ell}{\delta \nu} \diamond \nu_t = 0,
\\ \\&
\frac{\partial}{\partial t} \frac{\delta \ell}{\delta \nu} + u_t \star \frac{\delta \ell}{\delta \nu} - \frac{\delta \ell}{\delta
n}= 0,
\end{split}
\end{align}
with auxiliary equation 
\begin{equation} \label{eq:aux-thm}
\dot n_t = \nu_t + u_tn_t
,\end{equation}
obtained from the definitions 
$u_t = \dot g_tg_t^{-1}$ and $n_t = g_t\eta_t$, with $\nu_t = g_t \dot \eta_t$, provided the endpoint condition holds, that
\begin{equation} \label{eq:endptcond}
\frac{\delta \ell}{\delta u}(1) + \frac{\delta \ell}{\delta \nu}(1) \diamond n_1 = 0
,\end{equation}
at time $t=1$.
\end{enumerate}
\end{theorem}

\begin{corollary}[Coadjoint motion] \label{EPcorollary-coadjoint}
Equations \eqref{eq:meta.2-thm} and the auxiliary equation \eqref{eq:aux-thm} for $n_t$ together imply the following coadjoint motion equation,
\begin{align}
\frac{\partial}{\partial t} \Big(\frac{\delta \ell}{\delta u} +
\frac{\delta \ell}{\delta \nu}\diamond n\Big) 
+ \ad_{u_t}^*\Big(\frac{\delta \ell}{\delta u} +
\frac{\delta \ell}{\delta \nu}\diamond n\Big) = 0.
\label{coadjoint-mot1}
\end{align}

\end{corollary}

The equivalence of statements {\bf i} and {\bf ii} in Theorem \ref{equiv-thm-det} is classical, and no other proof will be offered here. 
The proofs of the other equivalences in Euler-Poincar\'e Theorem \ref{equiv-thm-det} and its Corollary \ref{EPcorollary-coadjoint} for deterministic metamorphosis are laid out in the sections below. 

\subsection{Deterministic Euler-Lagrange equations}

We compute the Euler-Lagrange equations associated with the minimization of the symmetry reduced action
$$
\mathcal{S} = 
\int_0^1 \ell(u_t, n_t, \nu_t) dt
$$
with fixed boundary conditions $n_0$ and $n_1$. We therefore consider
variations $\de u$ and $\om = \de n$. The variation $\de \nu$ can be
obtained from $n = g\eta$ and $\nu = g\dot \eta$ yielding $\dot n =
\nu + un$ and $\dot\om = \delta \nu + u\om + (\delta u) n$. Here and in
the following of this paper, we assume that computations are performed
in a local chart on $TN$ with respect to which we take partial derivatives.

We
therefore have
$$
\delta \mathcal{S} = \delta
\int_0^1 \left(\Lform{\frac{\delta \ell}{\delta u}}{ \delta u_t} + \Lform{
\frac{\delta \ell}{\delta n}}{\omega_t} + \Lform{\frac{\delta
\ell}{\delta \nu}}{\dot\omega_t - u_t\omega_t -(\delta u_t)n_t}\right) dt =0.
$$
The $\de u$ term yields the equation 
$$\frac{\delta \ell}{\delta u} + \frac{\delta \ell}{\delta \nu} \diamond n_t =
0.$$
where, in a slight abuse of notation,  $\de \ell/\de\nu \in T(TN)^*$ has been 
considered as a linear form on  $TN$ by $\lform{\de \ell/\de\nu}{z} :=
\lform{\de\ell/\de\nu}{(0, z)}$.
From the terms involving $\om$, we find, after an integration by parts
\begin{equation}
\label{eq:nu-var}
\frac{\partial}{\partial t} \frac{\delta
\ell}{\delta \nu} + u_t \star \frac{\delta
\ell}{\delta \nu} - \frac{\delta
\ell}{\delta n} = 0\,.
\end{equation}
Here, we have introduced notation for the star $(\star)$ operation,
\begin{equation}
\label{eq:star-def}
\Lform{\frac{\delta
\ell}{\delta \nu}}{u\om} =  \Lform{u \star \frac{\delta
\ell}{\delta \nu}}{\om}.
\end{equation}
For $u\om=\mathcal{L}_u\om$,
the star $(\star)$ operation denotes the dual of the Lie derivative, $u \star \nu^*=\mathcal{L}_u^T\nu^*$.

We therefore obtain the system of equations
\begin{align}
\label{eq:meta.1}
\begin{split}
\frac{\delta \ell}{\delta u} + \frac{\delta \ell}{\delta \nu} \diamond
n_t &= 0\,,\\ \\
\frac{\partial}{\partial t} \frac{\delta
\ell}{\delta \nu} + u_t \star \frac{\delta
\ell}{\delta \nu} &= \frac{\delta \ell}{\delta n}
\,,\\
\\
\dot n_t = \nu_t + u_t n_t\,.&
\end{split}
\end{align}
Note that the sum $\frac{\delta \ell}{\delta u} + \frac{\delta \ell}{\delta \nu} \diamond
n$ is the momentum map arising from Noether's theorem for the considered
invariance of the Lagrangian. The special form of the boundary
conditions (fixed $n_0$ and $n_1$) ensures that this momentum map vanishes.

\subsection{Deterministic Euler-Poincar\'e reduction}

As explained in in \cite{HoTrYo2009}, a system equivalent to that in \eqref{eq:meta.1} can be obtained via Euler-Poincar\'e
reduction \cite{HoMaRa1998}. In this setting, we make the
variation in the group element and in the template instead of the
velocity and the image. We denote $\xi_t  = \de g_t g_t^{-1}$ and
$\varpi_t = g_t\de\eta_t$. From these definitions, we obtain the expressions of the variations, $\de u$,
$\de n$ and $\de \nu$. 

We first have $\de u_t = \dot\xi_t + [\xi_t, u_t]$, which arises 
from the standard Euler-Poincar\'e reduction theorem, as
provided in \cite{HoMaRa1998,MR1994}. We also have $\de n_t = \de(g_t\eta_t)
= \varpi_t + \xi_t n_t$. 
From $\nu_t = g_t\dot\eta_t$, we get
$\de \nu_t = g_t\de\dot \eta_t + \xi_t\nu_t$ and from $\varpi_t = g_t\de\eta_t$ we also
have $\dot \varpi_t = u_t\varpi_t + g_t\dot\eta_t$. This yields
$\de\nu_t = \dot\varpi_t + \xi_t\nu_t - u_t\varpi_t.$

We also compute the boundary conditions for $\xi$ and $\varpi$. At
$t=0$, we have $g_0 = \id$ and $n_0 = g_0\eta_0 = \text{cst}$ which
implies $\xi_0=0$ and $\varpi_0=0$. At $t=1$, the relation $g_1\eta_1
= \text{cst}$ yields
$\xi_1 n_1 + \om_1 = 0$. 

Now, the first variation of  is
$$
\int_0^1 \left(\Lform{\frac{\delta \ell}{\delta u}}{\dot\xi_t - \ad_{u_t}\xi_t} 
+ \Lform{\frac{\delta \ell}{\delta n_t}}{\varpi_t + \xi_t n_t} 
+ \Lform{\frac{\delta\ell}{\delta \nu}}
{\dot\varpi_t + \xi_t\nu_t - u_t\varpi_t}\right) dt =0.
$$
In the integration by parts to eliminate $\dot\xi_t$ and $\dot\varpi_t$,
the boundary term is $\lform{(\de \ell/\de u)_1}{\xi_1} + \lform{(\de\ell
/ \de\nu)_1}{\om_1}$. Using the boundary condition, the last term can
be re-written
$$
- \lform{(\de\ell
/ \de\nu)_1}{\xi_1n_1} = \lform{(\de\ell
/ \de\nu)_1\diamond n_1}{\xi_1}.
$$

We therefore obtain the endpoint equation
$$
\frac{\delta \ell}{\delta u}(1) + \frac{\delta \ell}{\delta \nu}(1) \diamond
n_1 = 0.
$$
The evolution equation for ${\delta \ell}/{\delta u}$ is
$$
\frac{\partial}{\partial t} \frac{\delta \ell}{\delta u} + \ad^*_{u_t}
\frac{\delta \ell}{\delta u} + \frac{\delta \ell}{\delta n} \diamond n_t
+ \frac{\delta \ell}{\delta \nu} \diamond \nu_t = 0
,$$
and ${\delta \ell}/{\delta \nu}$ evolves by
$$
\frac{\partial}{\partial t} \frac{\delta \ell}{\delta \nu} +
u_t \star \frac{\delta \ell}{\delta \nu} - \frac{\delta \ell}{\delta
n} = 0.
$$
Consequently, we obtain the following system of equations,
\begin{align}
\label{eq:meta.2-redux}
\begin{split}
&\frac{\partial}{\partial t} \frac{\delta \ell}{\delta u} + \ad^*_{u_t}
\frac{\delta \ell}{\delta u} + \frac{\delta \ell}{\delta n} \diamond n_t
+ \frac{\delta \ell}{\delta \nu} \diamond \nu_t = 0,
\\ \\&
\frac{\partial}{\partial t} \frac{\delta \ell}{\delta \nu} + u_t \star \frac{\delta \ell}{\delta \nu} - \frac{\delta \ell}{\delta
n}= 0,
\end{split}
\end{align}
as well as the auxiliary equation 
\begin{equation} \label{eq:aux}
\frac{\partial}{\partial t} n_t = \nu_t + u_tn_t
,\end{equation}
obtained from the definitions 
$u_t = \dot g_tg_t^{-1}$ and $n_t = g_t\eta_t$. Moreover, the endpoint condition holds, that
\begin{equation} \label{eq:endptcond}
\frac{\delta \ell}{\delta u}(1) + \frac{\delta \ell}{\delta \nu}(1) \diamond n_1 = 0
, 
\end{equation}
at time $t=1$. \hfill $\Box$
\bigskip

As discussed in in \cite{HoTrYo2009}, the system \eqref{eq:meta.2-thm} is equivalent to
\eqref{eq:meta.1}, since they characterize the same critical
points. Direct evidence of this fact may be obtained from the proof of Corollary \ref{EPcorollary-coadjoint}, that
\begin{align}
\frac{\partial}{\partial t} \Big(\frac{\delta \ell}{\delta u} +
\frac{\delta \ell}{\delta \nu}\diamond n\Big) 
+ \ad_{u_t}^*\Big(\frac{\delta \ell}{\delta u} +
\frac{\delta \ell}{\delta \nu}\diamond n\Big) = 0.
\label{coadjoint-mot2}
\end{align}
\begin{proof}[Proof of Corollary \ref{EPcorollary-coadjoint}]
A solution of \eqref{eq:meta.2-thm} satisfies the coadjoint motion equation,
\begin{eqnarray*}
\frac{\partial}{\partial t}\Big(\frac{\delta \ell}{\delta u_t} + \frac{\delta \ell}{\delta \nu} \diamond
n_t\Big) &=& \frac{\partial}{\partial t} \frac{\delta \ell}{\delta u} +
\Big( \frac{\partial}{\partial t} \frac{\delta \ell}{\delta \nu}\Big)
\diamond n_t + \frac{\delta \ell}{\delta \nu} \diamond \dot n_t\\
&=&
\frac{\partial}{\partial t} \frac{\delta \ell}{\delta u} +
\Big( \frac{\delta \ell}{\delta n} - u_t \star \frac{\delta \ell}{\delta \nu} \Big)
\diamond n_t + \frac{\delta \ell}{\delta \nu} \diamond (\nu_t + u_tn_t)
\\
&=&
\frac{\partial}{\partial t} \frac{\delta \ell}{\delta u} +
\frac{\delta \ell}{\delta n} \diamond n_t + \frac{\delta \ell}{\delta
\nu} \diamond \nu_t -
\Big( u_t\star  \frac{\delta \ell}{\delta \nu} \Big)
\diamond n_t + \frac{\delta \ell}{\delta \nu} \diamond (u_tn_t)
\\
&=&
-\ad^*_{u_t} \frac{\delta \ell}{\delta u} - \ad^*_{u_t}(\frac{\delta \ell}{\delta
\nu} \diamond n_t).
\end{eqnarray*}
In the last equation, we have used the fact that, for any $\al\in\mathfrak g$, 
\begin{eqnarray*}
\Lform{\frac{\delta \ell}{\delta \nu} \diamond (un) - \Big( u\star \frac{\delta \ell}{\delta \nu} \Big)
\diamond n}{\al} &=& \Lform{\frac{\delta \ell}{\delta \nu}}{\al (un) -
u(\al n)} \\
&=& - \Lform{\frac{\delta \ell}{\delta \nu}}{[u, \al] n} \\
&=& - \Lform{\frac{\delta \ell}{\delta \nu} \diamond n}{[u, \al]} \\
&=& - \Lform{\ad^*_{u_t}(\frac{\delta \ell}{\delta
\nu} \diamond n_t)}{\al}.
\end{eqnarray*}
\end{proof}

\begin{remark}
Corollary \ref{EPcorollary-coadjoint} combined with the relation $(\delta \ell/\delta u)_1 +
(\delta \ell/\delta \nu)_1\diamond u_1 =0$ implies the first
equation in \eqref{eq:meta.1}. Namely, the zero level set of the momentum map is preserved by coadjoint motion. 
\end{remark}

\subsection{Deterministic Hamiltonian formulation}

The Euler-Poincar\'e formulation of metamorphosis in \eqref{eq:meta.2-thm} and \eqref{eq:aux-thm} in Theorem \ref{equiv-thm-det} allows passage to its Hamiltonian formulation via the following {\bfi Legendre transformation} of the reduced Lagrangian $\ell$ in the velocities ${u}$ and $\nu$, in the Eulerian fluid description,
\begin{equation}\label{liqxtal-legendre-xform}
\mu = \frac{\delta \ell}{\delta u}\,, \
\sigma = \frac{\delta \ell}{\delta \nu}\,, \quad
h(\mu, \sigma, n)
 = \lform{\mu}{u} + \lform{\sigma}{\nu}
- \ell(u, \nu, n).
\end{equation}
Accordingly, one computes the variational derivatives of $h$ as
\begin{equation}\label{liqxtal-dual-var-derivs}
\frac{\delta h}{\delta \mu} 
=
 u \,,
\quad
\frac{\delta h}{\delta \sigma} 
=
 \nu\,,
\quad
\frac{\delta h}{\delta n} 
=
- \,\frac{\delta \ell}{\delta n}\,.
\end{equation}
Consequently, the Euler-Poincar\'e equations \eqref{eq:meta.2-thm} and the auxiliary kinematic equation \eqref{eq:aux-thm} for metamorphosis imply the following equations, for the Legendre-transformed variables, $(\mu,\sigma, n)$,
\begin{align}
\label{eq:meta.2-redux}
\begin{split}
&{\partial_t}\mu 
+ \ad^*_{\delta h/\delta \mu} \mu 
+ \sigma \diamond \frac{\delta h}{\delta \sigma} 
- \frac{\delta h}{\delta n} \diamond n 
= 0,
\\ \\&
{\partial_t}\sigma + \mathcal{L}^T_{\delta h/\delta \mu}\sigma 
+ \frac{\delta h}{\delta n}= 0,
\end{split}
\end{align}
as well as the auxiliary equation 
\begin{equation} \label{eq:aux-met}
{\partial_t} n 
= \mathcal{L}_{\delta h/\delta \mu}n + \frac{\delta h}{\delta \sigma} 
.
\end{equation}
These equations are {\bfi Hamiltonian}. That is, they may be expressed in the
form 
\begin{equation}\label{Ham-form-met}
\frac{\partial \mathbf{z} }{\partial t}
=
\{\mathbf{z},h\} = \hbox{\sffamily b}\,\cdot
\frac{\delta h }{\delta \mathbf{z}}\,,
\end{equation}
where $\mathbf{z}\in (\mu, \sigma, n)$ and the Hamiltonian matrix
{\sffamily b} defines the Poisson bracket
\begin{equation} \label{PB-def-met}
\{f,h\}= \int d^{\,n\,}x \,
\frac{\delta f }{\delta \mathbf{z}}
\cdot\,\hbox{\sffamily b}\,\cdot
\frac{\delta h }{\delta \mathbf{z}}
\,,
\end{equation}
which is bilinear, skew symmetric and satisfies the {\bfi Jacobi identity},
\[
\{f,\{g,h\}\}+ \{g,\{h,f\}\}+ \{h,\{f,g\}\}= 0.
\] 
Assembling the metamorphosis equations \eqref{eq:meta.2-redux} -
\eqref{eq:aux-met} into the Hamiltonian form \eqref{Ham-form-met}
gives, 
\begin{equation}\label{Ham-matrix-met}
\begin{bmatrix}
\partial_t \mu
\\
\partial_t \sigma
\\
\partial_t n
\end{bmatrix}
=
-
\begin{bmatrix}
 \ad^*_{\Box} \mu & -\, \sigma  \diamond  \Box &  \Box \diamond n
\\
\mathcal{L}^T_{\Box} \sigma   &   0    &  -1
\\
- \,\mathcal{L}_{\Box}n           & 1      &   0
\end{bmatrix}
\begin{bmatrix}
{\delta h/\delta \mu}\,(=u)
\\
{\delta h/\delta \sigma}\,(=\nu)
\\
{\delta h/\delta n}
\end{bmatrix}
\end{equation}
In this expression, the operators act to the right on all terms in a product by the
chain rule.

{\bf Remarks about the Hamiltonian matrix.}
The Hamiltonian matrix in equation \eqref{Ham-matrix-met} was
discovered some time ago in the context of complex fluids in \cite{HK1988}. There, it was proven to be a valid Hamiltonian matrix by associating its Poisson bracket as defined in equation (\ref{PB-def-met}) with the dual space of a certain Lie algebra of
semidirect-product type which has a canonical two-cocycle on it. The
mathematical discussion of Lie algebras with two-cocycles is
given in \cite{HK1988}. See also \cite{Holm2002,GBR2009,GBT2010,GBHR2013} for further discussions of semidirect-product Poisson brackets with cocycles for complex fluids.

Being dual to the semidirect-product Lie algebra $\mathfrak{g}\,\circledS\, T^*N$, our Hamiltonian matrix in equation
(\ref{Ham-matrix-met}) is in fact a {\bfi Lie-Poisson
Hamiltonian matrix}. See, e.g., \cite{MR1994} and references
therein for more discussions of such Hamiltonian matrices. For our present
purposes, its rediscovery in the context of metamorphosis links the
physical and mathematical interpretations of the variables in the theory of
imaging science with earlier work in complex fluid dynamics and with the gauge theory approach to condensed matter, see, e.g., \cite{K1989}.

\subsection{Deterministic Hamilton-Pontryagin approach}

An alternative formulation to either the Euler-Lagrange equations, or the Euler-Poincar\'e approach is obtained in the Hamilton--Pontryagin principle. In this approach, the diffeomorphic paths appear explicitly. 
\begin{theorem}[Hamilton--Pontryagin principle]
\label{HPtheorem}
The Euler--Poincar\'e equations in Corollary \ref{EPcorollary-coadjoint} for coadjoint motion given by 
\begin{align}\label{EP-eqn-HP}
\begin{split}
&\frac{\partial}{\partial t} \Big(\frac{\delta \ell}{\delta u} +
\frac{\delta \ell}{\delta \nu}\diamond n\Big) 
+ \ad_{u}^*\Big(\frac{\delta \ell}{\delta u} +
\frac{\delta \ell}{\delta \nu}\diamond n\Big) = 0,
\\ \\&
\frac{\partial}{\partial t} \frac{\delta \ell}{\delta \nu} + u \star \frac{\delta \ell}{\delta \nu} - \frac{\delta \ell}{\delta n}= 0,
\end{split}
\end{align}
as well as the auxiliary equation 
\begin{equation} \label{eq:aux}
\dot n = \nu + un
,\end{equation}
on the space $\mathfrak{g}^*\times T^*N\times TN$ are equivalent to the following implicit variational principle,
\begin{equation}
\delta S(u,n,\dot{n},\nu,g,\dot{g})
= 0,
\label{HPimplicitvarprinc}
\end{equation}
for a constrained  action 
\begin{align}
S(u,n,\dot{n},\nu,g,\dot{g})
=
\int^1_0 \Big[\ell (u,n,\nu) 
+ 
\langle\, M\,,\,(\dot{g}g^{-1} - u)\,\rangle
+ 
\langle\, \sigma\,,\,(\dot n - \nu - un
)\,\rangle
\Big]\, dt
\,.
\label{implicit-actionHP}
\end{align}
\end{theorem}

\begin{proof}
The variations of $S$ in formula (\ref{implicit-actionHP}) are given by
\begin{align}
\begin{split}
0 = \delta S 
=&
\int^1_0
\Big\langle\,\frac{\delta \ell }{\delta u} 
- M + \sigma\diamond n\,,\, \delta u\,\Big\rangle
-
\Big\langle\, \dot{\sigma} + u\star \sigma - \frac{\delta\ell}{\delta n}  \,,\, \delta n
\,\Big\rangle
\\&\qquad
+
\Big\langle\,\frac{\delta \ell }{\delta \nu} 
- \sigma\,,\, \delta \nu\,\Big\rangle
+
\Big\langle\, M  \,,\, \delta (g^{-1}\dot{g})\,\Big\rangle
\, dt
\,.
\end{split}
\label{implicit-actionHP}
\end{align}
After a side calculation, one finds $\delta(\dot{g}g^{-1})=\dot{\xi} - {\rm ad}_u \xi$, with $\xi=\delta gg^{-1}$ for the last term in \eqref{implicit-actionHP}. Then, integrating by parts yields the familiar relation
\begin{eqnarray*}
\int^1_0
\Big\langle\, M  \,,\, \delta (\dot{g}g^{-1})
\,\Big\rangle
\, dt
&=&
\int_0^1
\Big\langle\, M  \,,\, (\dot{\xi} - {\rm ad}_u \xi) 
\,\Big\rangle
\, dt
\\&=&
\int_0^1
\Big\langle\, -\dot{M} -  {\rm ad}^*_u\,M \,,\, \eta
\,\Big\rangle
\, dt
+
\Big\langle\, M \,,\, \xi \,\Big\rangle\Big|_0^1
\,,
\end{eqnarray*}
where $\xi=\delta gg^{-1}$ vanishes at the endpoints in time. 
\smallskip

Thus, stationarity $\delta S=0$ of the Hamilton--Pontryagin variational
principle with constrained action integral \eqref{implicit-actionHP} yields
the following  set of equations: 
\begin{equation}
M 
= \frac{\delta \ell }{\delta u} + \sigma\diamond n
\,,\quad
\sigma = \frac{\delta \ell }{\delta \nu} 
\,,\quad
\frac{\partial\sigma}{\partial t} + u\star \sigma - \frac{\delta\ell}{\delta n}  = 0
\,,\quad
\frac{\partial M}{\partial t} + {\rm ad}^*_u\,M = 0
\,,
\label{symmHPeqns}
\end{equation}
as well as the constraint equations
\begin{equation}
\dot{g}g^{-1} - u = 0
\quad\hbox{and}\quad
\dot n - \nu - un = 0
\,.
\label{constraintHPeqns}
\end{equation}
This finishes the proof of the Hamilton--Pontryagin principle in Theorem \ref{HPtheorem}.
\end{proof}

\begin{proposition}[Untangling the Lie-Poisson structure \eqref{Ham-matrix-met}]$\,$\\
By the change of variables 
\begin{equation}\label{Ham-change of variables}
h(\mu,\sigma,n) = H(M,\sigma,n)\,,
\end{equation}
the Lie-Poisson structure \eqref{Ham-matrix-met} transforms into
\begin{equation}\label{Ham-matrix-met-untangled}
\begin{bmatrix}
\partial_t M
\\
\partial_t \sigma
\\
\partial_t n
\end{bmatrix}
=
-
\begin{bmatrix}
 {\rm ad}^*_{\Box} M & 0 & 0
\\
0   &   0    &  1
\\
0    & -1    &   0
\end{bmatrix}
\begin{bmatrix}
{\delta H/\delta M}\,(=u)
\\
{\delta H/\delta \sigma}+ \mathcal{L}_{\delta H/\delta M}n\,(=\nu+ \mathcal{L}_un)
\\
-{\delta H/\delta n} - \mathcal{L}^T_{\delta H/\delta M}\sigma\,
(=-{\delta H/\delta n} - \mathcal{L}^T_u\sigma)
\end{bmatrix},
\end{equation}
and thereby recovers equations \eqref{EP-eqn-HP} and \eqref{eq:aux} in Hamiltonian form. 
\end{proposition}
\begin{proof}
The proof follows from the expanding out the change of variables formula for variational derivatives,
\[
\delta h(\mu,\sigma,n) = \delta H(M,\sigma,n)\,,
\]
where $(M,\sigma,n) := (\mu + \sigma\diamond n,\sigma,n)$. Namely, one substitutes the corresponding terms,
\begin{align*}
\delta h(\mu,\sigma,n) &= \Lform{\frac{\delta h}{\delta \mu}}{\delta \mu}
+ \Lform{\frac{\delta h}{\delta \sigma}}{\delta \sigma} 
+ \Lform{\frac{\delta h}{\delta n}}{\delta n}
,\\
\delta H(M,\sigma,n) &= \Lform{\frac{\delta H}{\delta M}}{\delta \mu}
+ \Lform{\frac{\delta H}{\delta \sigma}+ \mathcal{L}_{\frac{\delta H}{\delta M}}n}{\delta \sigma}
- \Lform{\frac{\delta H}{\delta n} + \mathcal{L}^T_{\frac{\delta H}{\delta M}}\sigma}{\delta n},
\end{align*}
into the transformed Hamiltonian structure.
\end{proof}
\begin{remark}
The Lie-Poisson Hamiltonian structure \eqref{Ham-matrix-met-untangled} is the variable transformation \eqref{Ham-change of variables} of the corresponding structure \eqref{Ham-matrix-met}. The corresponding Lie-Poisson bracket is defined on the dual Lie algebra of the vector fields over the domain, $\mathcal{D}$; namely, $\mathfrak{g}=\mathfrak{X}(\mathcal{D})$ with canonical 2-cocycle $\mathfrak{X}(\mathcal{D})^*\times T^*\mathcal{F}(\mathcal{D},N)$, where $\mathcal{F}(\mathcal{D},N)$ denotes smooth functions from the domain, $\mathcal{D}$, to the data structure manifold, $N$.
For more details about how the untangling of Lie-Poisson structures is applied in geometric mechanics in the theory of complex fluids and for further citations in this literature to earlier work, see \cite{GBT2010,GBHR2013}.
\end{remark}

\section{Metamorphosis by stochastic Lie transport} \label{stoch-met-review}

\subsection{Notation and approach for the stochastic case}

To derive the stochastic partial differential equations (SPDEs) for uncertainty quantification in the metamorphosis approach to imaging science, we combine the recent developments for uncertainty quantification in fluid dynamics in \cite{holm2015variational} with the Hamilton-Pontryagin principle for metamorphosis discussed in the previous section. The idea is to replace the deterministic evolutionary constraints in equation \eqref{constraintHPeqns} by the following stochastic processes,
\begin{equation}
\dd{g}g^{-1} - \left(u_t(x)\,dt + \sum_i \xi_i(x)\circ dW^i_t \right) = 0
\quad\hbox{and}\quad
\dd n_t - \nu_t\,dt - \left(u_t\,dt + \sum_i \xi_i\circ dW^i_t \right)n = 0
\,,
\label{constraintHP-stoch}
\end{equation}
where $\dd$ is brief notation for the stochastic evolution operator, which strictly speaking is an integral operator. 
The first of these stochastic processes may be written equivalently as a stochastic version of the Lagrange-to-Euler map by using the notation $g_t^*$ for pullback by the stochastic diffeomorphism $g_t$, 
\begin{equation}
\dd{g}_t - g_t^* \left(u_t(X)\,dt + \sum_i \xi_i(X)\circ dW^i_t \right) = 0
\,.
\label{constraintHP-stoch-Lag}
\end{equation}
In this form, one sees that $g_t$ is a stochastic process with time dependent drift term given by the pullback operation, $u_t(x)\,dt = g_t^* u_t(X)\,dt = u_t(g_tX)\,dt$, in which subscript $t$ on $g_t$ and $u_t(X)$ indicates that both $u_t$ and $g_t$ depend explicitly on time, $t$. Thus, the dynamical drift velocity $u_t(x)$ depends on time explicitly and also through the Lagrange-to-Euler map $x=g_tX$ governed by \eqref{constraintHP-stoch-Lag} with initial value $g_0=Id$. The Lagrange-to-Euler map in \eqref{constraintHP-stoch-Lag} also contains a Stratonovich stochastic term, comprising a finite sum over time independent spatial functions $\xi_i$, $i=1,2,\dots,N$, each composed in a Stratonovich sense (denoted by the symbol $\circ$) with its own Brownian motion in time, $dW^i_t$. This type of Stratonovich stochasticity, called ``cylindrical noise'', was introduced in \cite{Sc1988}. In the cylindrical noise term, the $\xi_i(x)$, $i=1,2,\dots,N$, are interpreted as describing the spatial correlations of the noise in fixed Eulerian space, e.g., as eigenvectors of the correlation tensor, or covariance, for a process which is assumed to be statistically stationary. 

\subsection{Stochastic Hamilton-Pontryagin approach} 

\begin{theorem}[Stochastic Hamilton--Pontryagin principle]
\label{HPtheorem-stoch}
Stochastic metamorphosis is governed by coadjoint motion represented as SPDE given by 
\begin{align}\label{EP-eqn-HP-stoch}
\begin{split}
&\dd\Big(\frac{\delta \ell}{\delta u} +
\frac{\delta \ell}{\delta \nu}\diamond n\Big) 
+ \ad_{\left(u_t(x)\,dt + \sum_i \xi_i(x)\circ dW^i_t \right)}^*\Big(\frac{\delta \ell}{\delta u} +
\frac{\delta \ell}{\delta \nu}\diamond n\Big) = 0,
\\ \\&
\dd \frac{\delta \ell}{\delta \nu} 
+ \left(u_t(x)\,dt + \sum_i \xi_i(x)\circ dW^i_t \right) \star \frac{\delta \ell}{\delta \nu} 
- \frac{\delta \ell}{\delta n}= 0,
\end{split}
\end{align}
as well as the auxiliary equation 
\begin{equation} \label{eq:aux-stoch}
\dd n = \nu + \left(u_t(x)\,dt + \sum_i \xi_i(x)\circ dW^i_t \right)n
,\end{equation}
on the space $\mathfrak{g}^*\times T^*N\times TN$ are equivalent to the following implicit variational principle,
\begin{equation}
\delta S(u,n,\dd{n},\nu,g,\dd{g})
= 0,
\label{HPimplicitvarprinc}
\end{equation}
for a stochastically constrained  action 
\begin{align}
\begin{split}
S(u,n,\dd{n},\nu,g,\dd{g})
&=
\int^1_0 \Big[\ell (u_t,n_t,\nu_t) \, dt
+ 
\Big\langle\, M\,,\,\dd{g}_tg_t^{-1} - u_t(x)\,dt - \sum_i \xi_i(x)\circ dW^i_t \,\Big\rangle
\\  &\qquad + 
\Big\langle\, \sigma\,,\,(\dd n_t - \nu_t\,dt - \Big(u_t(x)\,dt + \sum_i \xi_i(x)\circ dW^i_t \Big)n_t
)\,\Big\rangle
\Big]
\,.
\end{split}
\label{implicit-actionHP-stoch}
\end{align}
\end{theorem}

\begin{remark}[Stratonovich versus It\^o representations]
In dealing with the stochastic variational principle, we will work in the Stratonovich representation, because it admits ordinary variational calculus. However, later, when we consider expected values for the solutions, we will transform to the equivalent It\^o representation. In transforming to the It\^o representation, we will discover that the effective diffusion from the It\^o contraction term is by no means a Laplacian. Instead, 
the It\^o contraction term turns out to produce a double Lie derivative with respect to the sum of vector fields $\xi_i(x)$.
\end{remark}

After this remark, we return to the proof of Theorem \ref{HPtheorem-stoch} for the Stochastic Hamilton--Pontryagin principle.
\begin{proof}
The variations of $S$ in formula (\ref{implicit-actionHP-stoch}) are given by
\begin{align}
\begin{split}
0 = \delta S 
=&
\int^1_0
\Big\langle\,\frac{\delta \ell }{\delta u} 
- M + \sigma\diamond n\,,\, \delta u\,\Big\rangle\,dt
-
\Big\langle\, \dd{\sigma} + \Big(u_t(x)\,dt + \sum_i \xi_i(x)\circ dW^i_t \Big)\star \sigma 
- \frac{\delta\ell}{\delta n}\,dt  \,,\, \delta n
\,\Big\rangle
\\&\qquad
+
\Big\langle\,\frac{\delta \ell }{\delta \nu} 
- \sigma\,,\, \delta \nu\,\Big\rangle\,dt
+
\Big\langle\, M  \,,\, \delta (\dd{g}g^{-1})\,\Big\rangle
\, dt
\,.
\end{split}
\label{implicit-actionHP}
\end{align}
In a side calculation, one finds 
\[
\delta(\dd{g}g^{-1})=\dd{\xi} - {\rm ad}_{(u_t(x)\,dt + \sum_i \xi_i(x)\circ dW^i_t)} \xi\,, \quad\hbox{with}\quad 
\xi=\delta gg^{-1}
\,,
\]
for substitution into the last term. Then, integrating by parts yields the relation
\begin{eqnarray*}
\int^1_0
\Big\langle\, M  \,,\, \delta (\dd{g}g^{-1})
\,\Big\rangle
\, dt
&=&
\int_0^1
\Big\langle\, M  \,,\, \dd{\xi} - {\rm ad}_{(u_t(x)\,dt + \sum_i \xi_i(x)\circ dW^i_t)} \xi
\,\Big\rangle
\, dt
\\&=&
\int_0^1
\Big\langle\, -\dd{M} -  {\rm ad}^*_{(u_t(x)\,dt + \sum_i \xi_i(x)\circ dW^i_t)}\,M \,,\, \eta
\,\Big\rangle
\, dt
+
\Big\langle\, M \,,\, \xi \,\Big\rangle\Big|_0^1
\,,
\end{eqnarray*}
where $\xi=\delta gg^{-1}$ vanishes at the endpoints in time. 
\smallskip

Thus, stationarity $\delta S=0$ of the Hamilton--Pontryagin variational
principle with stochastically constrained action integral \eqref{implicit-actionHP} yields
the following  set of SPDEs: 
\begin{align}
\begin{split}
&
\dd M + {\rm ad}^*_{(u_t(x)\,dt + \sum_i \xi_i(x)\circ dW^i_t)}\,M = 0
\,,\\ &
\dd \sigma
+ \Big(u_t(x)\,dt + \sum_i \xi_i(x)\circ dW^i_t\Big)\star \sigma 
- \frac{\delta\ell}{\delta n} \,dt = 0 
\,,
\end{split}
\label{symmHPeqns-stoch}
\end{align}
for the quantities
\begin{align}
M 
= \frac{\delta \ell }{\delta u} + \sigma\diamond n
\,,\quad\hbox{and}\quad
\sigma = \frac{\delta \ell }{\delta \nu} 
\,,
\label{mugammaeqns-stoch}
\end{align}
as well as the stochastic constraint equations
\begin{align}
\begin{split}
&\dd{g}g^{-1} - \Big(u_t(x)\,dt + \sum_i \xi_i(x)\circ dW^i_t\Big) = 0\,,
\\&
\dd n - \nu - \Big(u_t(x)\,dt + \sum_i \xi_i(x)\circ dW^i_t\Big)n = 0
\,.
\end{split}
\label{constraintHPeqns-stoch}
\end{align}
This finishes the proof of the Hamilton--Pontryagin principle for stochastic metamorphosis formulated in Theorem \ref{HPtheorem-stoch}.
\end{proof}

\subsection{Stochastic Hamiltonian formulation}

By Corollary  \ref{EPcorollary-coadjoint}, the stochastic equations \eqref{symmHPeqns-stoch} through \eqref{constraintHPeqns-stoch} above imply the corresponding stochastic versions of  in \eqref{eq:meta.2-thm} and \eqref{eq:aux-thm} in Theorem \ref{equiv-thm-det}. Namely,
\begin{align}
\label{eq:meta.2-thm-stoch}
\begin{split}
&\dd \frac{\delta \ell}{\delta u} + \ad^*_{\tilde u}
\frac{\delta \ell}{\delta u} + \frac{\delta \ell}{\delta n} \diamond n\,dt
+ \frac{\delta \ell}{\delta \nu} \diamond \nu\,dt = 0,
\\ \\&
\dd\frac{\delta \ell}{\delta \nu} + \mathcal{L}^T_{\tilde u} \frac{\delta \ell}{\delta \nu} - \frac{\delta \ell}{\delta n}\,dt
= 0,
\\ \\&
\dd n =  \mathcal{L}_{\tilde u}n_t + \nu \,dt
\end{split}
\end{align}
with Stratonovich stochastic transport velocity $\tilde u$ given by
\begin{equation} \label{eq:trans-stoch-vel}
\tilde u := u_t(x)\,dt + \sum_i \xi_i(x)\circ dW^i_t
\,.
\end{equation}

At this point the Hamiltonian structure of the deterministic metamorphosis equations reveals how we can write the stochastic metamorphosis equations in Hamiltonian form. Namely, we deform the deterministic Hamiltonian by adding the stochastic part 
to it as being linear in the momentum $\mu$ and paired with the Stratonovich noise perturbation, 
\begin{equation} \label{eq:trans-stoch-Ham}
\tilde h := h(\mu. \nu, n)\,dt + \Big\langle \mu\,,\,\sum_i \xi_i(x)\circ dW^i_t\Big\rangle
\,.
\end{equation}
We then use the same Lie-Poisson Hamiltonian structure as in the deterministic case. 

Assembling the metamorphosis equations \eqref{eq:meta.2-thm-stoch} into the Hamiltonian form \eqref{Ham-form-met}
gives, 
\begin{equation}\label{Ham-matrix-met-stoch}
\begin{bmatrix}
\dd \mu
\\
\dd \sigma
\\
\dd n
\end{bmatrix}
=
-
\begin{bmatrix}
 \ad^*_{\Box} \mu & -\, \sigma  \diamond  \Box &  \Box \diamond n
\\
\mathcal{L}^T_{\Box} \sigma   &   0    &  -1
\\
- \,\mathcal{L}_{\Box}n           & 1      &   0
\end{bmatrix}
\begin{bmatrix}
{\delta \tilde{h}/\delta \mu}\,(=\tilde u := u_t(x)\,dt + \sum_i \xi_i(x)\circ dW^i_t)
\\
{\delta \tilde{h}/\delta \sigma} = {\delta h/\delta \sigma}\,(=\nu)
\\
{\delta \tilde{h}/\delta \sigma} = {\delta h/\delta n}
\end{bmatrix}
.\end{equation}
As before, the operators in the Hamiltonian matrix act to the right on all terms in a product by the chain rule.

By the change of variables corresponding to \eqref{Ham-change of variables} for this stochastic case,
\begin{equation}\label{Ham-change of variables-stoch}
\tilde{h}(\mu,\sigma,n) = \widetilde{H}(M,\sigma,n)
= {H}(M,\sigma,n)\,dt + + \Big\langle M\,,\,\sum_i \xi_i(x)\circ dW^i_t\Big\rangle
\,,
\end{equation}
one finds that the Lie-Poisson structure in \eqref{Ham-matrix-met-stoch} transforms into
\begin{equation}\label{Ham-matrix-met-untangled-stoch}
\begin{bmatrix}
\dd M
\\
\dd  \sigma
\\
\dd n
\end{bmatrix}
=
-
\begin{bmatrix}
 {\rm ad}^*_{\Box} M & 0 & 0
\\
0   &   0    &  1
\\
0    & -1    &   0
\end{bmatrix}
\begin{bmatrix}
{\delta \widetilde{H}/\delta M}\,(=\widetilde{u})
\\
{\delta \widetilde{H}/\delta \sigma}+ \mathcal{L}_{\delta \widetilde{H}/\delta M}n
\,(=\nu+ \mathcal{L}_{\widetilde{u}}n)
\\
-{\delta \widetilde{H}/\delta n} - \mathcal{L}^T_{\delta \widetilde{H}/\delta M}\sigma\,
(=-{\delta \widetilde{H}/\delta n} - \mathcal{L}^T_{\widetilde{u}}\sigma)
\end{bmatrix},
\end{equation}
and thereby recovers equations \eqref{symmHPeqns-stoch} - \eqref{constraintHPeqns-stoch} in Hamiltonian form. 

\begin{remark}
The advantage of writing equations \eqref{Ham-matrix-met-untangled-stoch} in terms of the total momentum 1-form density $M:=\mathbf{M}\cdot d\mathbf{x}\otimes d^3x$ may be seen by recalling that ${\rm ad}^*_{\widetilde{u}} M = \mathcal{L}_{\widetilde{u}} M$ for 1-form densities. Consequently, the first equation in \eqref{Ham-matrix-met-untangled-stoch} implies $(\dd + \mathcal{L}_{\widetilde{u}})M=0$, which in turn implies that
\begin{equation}\label{Advect-M}
\dd (g^*M) = g^* \Big((\dd + \mathcal{L}_{\widetilde{u}})M\Big) = 0\,.
\end{equation} 
This means  the total momentum 1-form density $M:=\mathbf{M}\cdot d\mathbf{x}\otimes d^3x$ is preserved by the stochastic flow given by the pullback of the stochastic diffeomorphism $g_t$ in \eqref{constraintHP-stoch-Lag}, which is the flow of the stochastic vector field $\widetilde{u}$. That is, the stochastic Lagrange-to-Euler flow $g_t$, which is the solution of $\dd g g^{-1}=\widetilde{u}$ in \eqref{constraintHP-stoch}, 
\begin{equation}
\dd{g}_t - g_t^*\widetilde{u}
=
\dd{g}_t - g_t^* \left(u_t(X)\,dt + \sum_i \xi_i(X)\circ dW^i_t \right) = 0
\,,
\label{constraintHP-stoch-Lag-redux}
\end{equation} 
preserves the quantity $M$ along its flow. Hence, we say that the total momentum 1-form density $M:=\mathbf{M}\cdot d\mathbf{x}\otimes d^3x$ is \emph{stochastically advected}.
\end{remark}

{\bf Summary.} The preservation of the Hamiltonian structure achieved in \eqref{Ham-matrix-met-stoch} for the present formulation of the stochastic metamorphosis equations provides the interpretation of the stochastic part of the flow. The Hamiltonian flow of the momentum $\langle \mu\,,\,\sum_i \xi_i(x)\circ dW^i_t\rangle$ produces  stochastic translation in Eulerian space with velocity $\sum_i \xi_i(x)\circ dW^i_t$. Thus, adding the stochastic part, linear in the momentum, to the metamorphosis Hamiltonian $h(\mu,\nu,n)$ adds a stochastic transport to the deterministic flow. This is consistent with our intention of modelling stochastic metamorphosis as motion generated by a temporally stochastic flow on the diffeomorphisms, with spatial correlations given by the prescribed, time-independent correlation eigenvectors determined from data assimilation. 

\subsection{It\^o representation}
In preparation for writing the It\^o representation, we first substitute ${\rm ad}_u^*\mu=\mathcal{L}_u\mu$ to show in the more familiar Lie derivative notation the action of the vector field $u$ on its dual momentum, the 1-form density $\mu={\delta \ell}/{\delta u}$.
The equivalent It\^o representations of equations \eqref{eq:meta.2-thm-stoch} are then given by 
\begin{align}
\label{eq:meta.2-thm-stoch}
\begin{split}
&\dd \frac{\delta \ell}{\delta u} 
+ \mathcal{L}_{u} \frac{\delta \ell}{\delta u} \,dt
+ \sum_i \mathcal{L}_{\xi_i(x)} \frac{\delta \ell}{\delta u}  dW^i_t
-
\frac12 \sum_i \mathcal{L}_{\xi_i(x)}\Big(\mathcal{L}_{\xi_i(x)}
\frac{\delta \ell}{\delta u}\Big)\,dt
+ \frac{\delta \ell}{\delta n} \diamond n\,dt
+ \frac{\delta \ell}{\delta \nu} \diamond \nu\,dt = 0,
\\ \\&
\dd\frac{\delta \ell}{\delta \nu} 
+ \mathcal{L}^T_{u} \frac{\delta \ell}{\delta \nu} \,dt
+ \sum_i \mathcal{L}^T_{\xi_i(x)} \frac{\delta \ell}{\delta \nu}  dW^i_t
-
\frac12 \sum_i \mathcal{L}^T_{\xi_i(x)}\Big(\mathcal{L}^T_{\xi_i(x)}
\frac{\delta \ell}{\delta \nu}\Big)\,dt
- \frac{\delta \ell}{\delta n}\,dt
= 0,
\\ \\&
\dd n =   \mathcal{L}_{u} n_t\,dt
+ \sum_i \mathcal{L}_{\xi_i(x)} n_t \,dW^i_t
-
\frac12 \sum_i \mathcal{L}_{\xi_i(x)}\Big(\mathcal{L}_{\xi_i(x)}
n_t\Big)\,dt
+ \nu \,dt,
\end{split}
\end{align}
with stochastic transport velocity $\widehat u$ given in It\^o form by
\begin{equation} \label{eq:trans-stoch-vel}
\widehat u := u_t(x)\,dt + \sum_i \xi_i(x) dW^i_t
\,,
\end{equation}
plus the It\^o contraction drift terms. Likewise, the stochastic advection of the total momentum 1-form density $M = \mu + \sigma\diamond n =\mathbf{M}\cdot d\mathbf{x}\otimes d^3x$, expressed in Stratonovich form as $(\dd + \mathcal{L}_{\widetilde{u}})M=0$, is expressed in It\^o form as
\begin{equation}
\label{eq:M-stochadvect}
\dd M 
+ \mathcal{L}_{u} M \,dt
+ \sum_i \mathcal{L}_{\xi_i(x)} M  dW^i_t
-
\frac12 \sum_i \mathcal{L}_{\xi_i(x)}\Big(\mathcal{L}_{\xi_i(x)}
M\Big)\,dt = 0\,,
\end{equation}
In which the last sum is the It\^o contraction term. 

In It\^o form, the expectation of the noise terms vanish. The noise interacts multiplicatively with both the solution $ \mathbf{M}$ and the gradient of the solution $\nabla  \mathbf{M}$,  through the Lie derivative, as
\begin{equation}
\label{eq:M-stochad-vector}
\sum_i \mathcal{L}_{\xi_i(x)} M  dW^i_t = 
\sum_i \Big\{\Big[ \big(\bsxi_i(\bx)\cdot\nabla\big) \mathbf{M}  
+ \big(\nabla\bsxi_i(\bx)\big) ^T\cdot \mathbf{M}
+ \mathbf{M} \,{\rm div}\bsxi_i(\bx) 
\Big]dW^i_t\Big\} \cdot d\mathbf{x}\otimes d^3x
.\end{equation}
Likewise, the It\^o contraction drift terms are not Laplacians, as would have been the case for additive noise with constant amplitude. Instead, in \eqref{eq:M-stochadvect} they are sums over double Lie derivatives with respect to the vector fields $\xi_i(x)$ associated with the spatial correlations of the stochastic perturbation. This double Lie derivative combination was called the \emph{Lie Laplacian} in \cite{holm2015variational} has many properties of potential interest in the mathematical analysis of these equations \cite{CrFlHo2017}.

\subsection*{Acknowledgements.}
I am grateful to A. Trouv\'e and L. Younes for their collaboration in developing the Euler-Poincar\'e description of metamorphosis. I am also grateful to F. Gay-Balmaz, T. S. Ratiu and C. Tronci for many illuminating collaborations in complex fluids and other topics in geometric mechanics during the course of our long friendship. Finally, I am also grateful to M. I. Miller and D. Mumford for their encouragement over the years to pursue the role of EPDiff in imaging science. 
During the present work the author was  partially supported by the European Research Council Advanced Grant 267382 FCCA and the EPSRC Grant EP/N023781/1.

\bibliographystyle{plainnat}
\bibliography{defor,biblio,ss}

\begin{thebibliography}{35}
\providecommand{\natexlab}[1]{#1}
\providecommand{\url}[1]{\texttt{#1}}
\expandafter\ifx\csname urlstyle\endcsname\relax
  \providecommand{\doi}[1]{doi: #1}\else
  \providecommand{\doi}{doi: \begingroup \urlstyle{rm}\Url}\fi

\bibitem[Arnaudon et~al.(2017{\natexlab{a}})Arnaudon, Holm, and
  Sommer]{AHS2017}
A.~Arnaudon, D.~D. Holm, and S.~Sommer.
\newblock A geometric framework for stochastic shape analysis.
\newblock \emph{arXiv preprint arXiv:1703.0997}, 2017{\natexlab{a}}.

\bibitem[Arnaudon et~al.(2016)Arnaudon, Castro, and Holm]{arnaudon2016noise}
Alexis Arnaudon, Alex~L Castro, and Darryl~D Holm.
\newblock Noise and dissipation on coadjoint orbits.
\newblock \emph{arXiv preprint arXiv:1601.02249}, 2016.

\bibitem[Arnaudon et~al.(2017{\natexlab{b}})Arnaudon, Holm, Pai, and
  Sommer]{arnaudon2016stochastic2}
Alexis Arnaudon, Darryl~D Holm, Akshay Pai, and Stefan Sommer.
\newblock A stochastic large deformation model for computational anatomy.
\newblock In \emph{Information {Processing} for {Medical} {Imaging} ({IPMI})},
  2017{\natexlab{b}}.

\bibitem[Beg et~al.(2005)Beg, Miller, Trouv{\'e}, and Younes]{beg2005computing}
M~Faisal Beg, Michael~I Miller, Alain Trouv{\'e}, and Laurent Younes.
\newblock Computing large deformation metric mappings via geodesic flows of
  diffeomorphisms.
\newblock \emph{International journal of computer vision}, 61\penalty0
  (2):\penalty0 139--157, 2005.

\bibitem[Bruveris et~al.(2011)Bruveris, Gay-Balmaz, Holm, and
  Ratiu]{bruveris2011momentum}
Martins Bruveris, Fran{\c{c}}ois Gay-Balmaz, Darryl~D Holm, and Tudor~S Ratiu.
\newblock The momentum map representation of images.
\newblock \emph{Journal of Nonlinear Science}, 21\penalty0 (1):\penalty0
  115--150, 2011.

\bibitem[Christensen et~al.(1996)Christensen, Rabbitt, and
  Miller]{christensen_deformable_1996}
Gary~E. Christensen, Richard Rabbitt, and Michael~I. Miller.
\newblock Deformable templates using large deformation kinematics.
\newblock \emph{Image Processing, IEEE Transactions on}, 5\penalty0 (10), 1996.

\bibitem[Crisan et~al.(2017)Crisan, Flandoli, and Holm]{CrFlHo2017}
D.~O. Crisan, F.~Flandoli, and D.~D. Holm.
\newblock Solution properties of a 3{D} stochastic {E}uler fluid equation.
\newblock \emph{arXiv preprint arXiv:1704.06989}, 2017.
\newblock URL \url{https://arxiv.org/abs/1704.06989}.

\bibitem[Dupuis et~al.(1998)Dupuis, Grenander, and
  Miller]{dupuis_variational_1998}
Paul Dupuis, Ulf Grenander, and Michael~I. Miller.
\newblock Variational {Problems} on {Flows} of {Diffeomorphisms} for {Image}
  {Matching}.
\newblock \emph{Quarterly of applied mathematics}, 1998.

\bibitem[Gay-Balmaz and Ratiu(2009)]{GBR2009}
F.~Gay-Balmaz and T.~S. Ratiu.
\newblock The geometric structure of complex fluids.
\newblock \emph{Advances in Applied Mathematics}, 42\penalty0 (12):\penalty0
  176--275, 2009.

\bibitem[Gay-Balmaz and Tronci(2010)]{GBT2010}
F.~Gay-Balmaz and C.~Tronci.
\newblock Reduction theory for symmetry breaking with applications to nematic
  systems.
\newblock \emph{Physica D}, 239\penalty0 (20-22):\penalty0 1929--1947, 2010.

\bibitem[Gay-Balmaz et~al.(2013)Gay-Balmaz, Holm, and Ratiu]{GBHR2013}
F.~Gay-Balmaz, D.~D. Holm, and T.~S. Ratiu.
\newblock Geometric dynamics of optimization.
\newblock \emph{Comm. in Math. Sciences}, 11\penalty0 (1):\penalty0 163--231,
  2013.

\bibitem[Grenander(1994)]{grenander_general_1994}
Ulf Grenander.
\newblock \emph{General {Pattern} {Theory}: {A} {Mathematical} {Study} of
  {Regular} {Structures}}.
\newblock Oxford University Press, USA, February 1994.
\newblock ISBN 0-19-853671-2.

\bibitem[Holm(2002)]{Holm2002}
D.~D. Holm.
\newblock Euler-{P}oincar\'e dynamics of perfect complex fluids.
\newblock In P.~Newton, P.~Holmes, and A.~Weinstein, editors, \emph{Geometry,
  Mechanics, and Dynamics: in honor of the 60th birthday of Jerrold E.
  Marsden}, pages 113--167. Springer, 2002.

\bibitem[Holm and Kupershmidt(1988)]{HK1988}
D.~D. Holm and B.~Kupershmidt.
\newblock The analogy between spin glasses and {Y}ang-{M}ills fluids.
\newblock \emph{Journal of Mathematical Physics}, 29:\penalty0 21--30, 1988.

\bibitem[Holm et~al.(1998{\natexlab{a}})Holm, Marsden, and Ratiu]{HoMaRa1998}
D.~D. Holm, J.~E. Marsden, and T.~S. Ratiu.
\newblock The {Euler--Poincar\'e} equations and semidirect products with
  applications to continuum theories.
\newblock \emph{Adv. in Math.}, 137:\penalty0 1--81, 1998{\natexlab{a}}.

\bibitem[Holm et~al.(2009)Holm, Trouv\'e, and Younes]{HoTrYo2009}
D.~D. Holm, A.~Trouv\'e, and L.~Younes.
\newblock The {E}uler-{P}oincar\'e theory of metamorphosis.
\newblock \emph{Quarterly of Applied Mathematics}, 67:\penalty0 661--685, 2009.

\bibitem[Holm(2011)]{holm_geometric_2011}
Darryl~D. Holm.
\newblock \emph{Geometric {Mechanics} - {Part} {I}: {Dynamics} and {Symmetry}}.
\newblock Imperial College Press, London : Hackensack, NJ, 2 edition edition,
  2011.
\newblock ISBN 978-1-84816-775-9.

\bibitem[Holm(2015)]{holm2015variational}
Darryl~D Holm.
\newblock Variational principles for stochastic fluid dynamics.
\newblock \emph{Proceedings of the Royal Society of London A: Mathematical,
  Physical and Engineering Sciences}, 471\penalty0 (2176):\penalty0 20140963,
  2015.

\bibitem[Holm and Marsden(2005)]{HoMa2004}
Darryl~D. Holm and Jerrold~E Marsden.
\newblock Momentum maps and measure-valued solutions (peakons, filaments, and
  sheets) for the {EPD}iff equation.
\newblock In \emph{The {b}readth of {S}ymplectic and {P}oisson {G}eometry},
  pages 203--235. Springer, 2005.

\bibitem[Holm and Tyranowski(2016{\natexlab{a}})]{HoTy2016}
Darryl~D Holm and Tomasz~M Tyranowski.
\newblock Variational principles for stochastic soliton dynamics.
\newblock \emph{Proc. R. Soc. A}, 472\penalty0 (2187):\penalty0 20150827,
  2016{\natexlab{a}}.

\bibitem[Holm and Tyranowski(2016{\natexlab{b}})]{holm2016stochastic}
Darryl~D Holm and Tomasz~M Tyranowski.
\newblock Stochastic discrete hamiltonian variational integrators.
\newblock \emph{arXiv preprint arXiv:1609.00463}, 2016{\natexlab{b}}.

\bibitem[Holm et~al.(1998{\natexlab{b}})Holm, Marsden, and
  Ratiu]{holm1998euler}
Darryl~D Holm, Jerrold~E Marsden, and Tudor~S Ratiu.
\newblock The {E}uler--{P}oincar{\'e} equations and semidirect products with
  applications to continuum theories.
\newblock \emph{Advances in Mathematics}, 137\penalty0 (1):\penalty0 1 -- 81,
  1998{\natexlab{b}}.
\newblock \doi{http://dx.doi.org/10.1006/aima.1998.1721}.

\bibitem[Kleinert(1989)]{K1989}
H.~Kleinert.
\newblock \emph{Gauge {F}ields in {C}ondensed {M}atter}.
\newblock World Scientific, 1989.

\bibitem[Marsden and Ratiu(1994)]{MR1994}
J.~E. Marsden and T.~S. Ratiu.
\newblock \emph{Introduction to {M}echanics and {S}ymmetry}.
\newblock Springer-Verlag, 1994.

\bibitem[Marsland and Shardlow(2016)]{marsland2016langevin}
Stephen Marsland and Tony Shardlow.
\newblock Langevin equations for landmark image registration with uncertainty.
\newblock \emph{arXiv preprint arXiv:1605.09276}, 2016.

\bibitem[Miller et~al.(2015)Miller, Trouv\'e, and Younes]{MTY2015}
M.I. Miller, A.~Trouv\'e, and L.~Younes.
\newblock Hamiltonian systems and optimal control in computational anatomy: 100
  years since d'arcy thompson.
\newblock \emph{Annu. Rev. Biomed. Eng.}, 17:\penalty0 447--509, 2015.

\bibitem[Miller et~al.(2002)Miller, Trouv{\'e}, and Younes]{MiTrYo2002}
Michael~I Miller, Alain Trouv{\'e}, and Laurent Younes.
\newblock On the metrics and {E}uler--{L}agrange equations of computational
  anatomy.
\newblock \emph{Annual Review of Biomedical Engineering}, 4:\penalty0 375--405,
  2002.

\bibitem[Richardson and Younes(2016)]{RY2016}
C.~L. Richardson and L.~Younes.
\newblock Metamorphosis of images in reproducing kernel hilbert spaces.
\newblock \emph{Advances in Computational Mathematics}, 42\penalty0
  (3):\penalty0 573Ð603, 2016.

\bibitem[Schauml\"offel(1988)]{Sc1988}
K.-U. Schauml\"offel.
\newblock White noise in space and time and the cylindrical wiener process.
\newblock \emph{Stochastic Analysis and Applications}, 6\penalty0 (1):\penalty0
  81--89, 1988.

\bibitem[Trouv\'e and Younes(2005)]{ty05b}
A.~Trouv\'e and L.~Younes.
\newblock Metamorphoses through {L}ie group action.
\newblock \emph{Found. Comp. Math.}, pages 173--198, 2005.

\bibitem[Trouv{\'e}(1995)]{trouve_infinite_1995}
Alain Trouv{\'e}.
\newblock An infinite dimensional group approach for physics based models in
  pattern recognition.
\newblock \emph{preprint}, 1995.

\bibitem[Trouv{\'e} and Vialard(2012)]{TrVi2012}
Alain Trouv{\'e} and Fran\c{c}ois-Xavier Vialard.
\newblock Shape splines and stochastic shape evolutions: a second order point
  of view.
\newblock \emph{Quarterly of Applied Mathematics}, 70\penalty0 (2):\penalty0
  219--251, 2012.

\bibitem[Vialard(2013)]{vialard2013extension}
Fran{\c{c}}ois-Xavier Vialard.
\newblock Extension to infinite dimensions of a stochastic second-order model
  associated with shape splines.
\newblock \emph{Stochastic Processes and their Applications}, 123\penalty0
  (6):\penalty0 2110--2157, 2013.

\bibitem[Younes(2010)]{younes_shapes_2010}
Laurent Younes.
\newblock \emph{Shapes and {Diffeomorphisms}}.
\newblock Springer, 2010.
\newblock ISBN 978-3-642-12054-1.

\bibitem[Younes et~al.(2009)Younes, Arrate, and Miller]{younes_evolutions_2009}
Laurent Younes, Felipe Arrate, and Michael~I. Miller.
\newblock Evolutions equations in computational anatomy.
\newblock \emph{NeuroImage}, 45\penalty0 (1, Supplement 1):\penalty0 S40--S50,
  2009.
\newblock ISSN 1053-8119.
\newblock \doi{10.1016/j.neuroimage.2008.10.050}.

\end{thebibliography}

\end{document}